\documentclass[10pt, reqno]{amsart} 

 	\usepackage{setspace}
	\usepackage{indentfirst} 
	\usepackage{amsmath}
	\usepackage{amssymb}
	\usepackage{amsfonts}
	\usepackage{verbatim}
	\usepackage[mathscr]{eucal}
	\usepackage{amsthm}
	\usepackage[numbers]{natbib} 
	
\theoremstyle{plain} 
\newtheorem{theorem}{Theorem}[section]
\newtheorem{lemma}{Lemma}[section]

\theoremstyle{remark}
\newtheorem{remark}{Remark}[section]

	\renewcommand{\MR}[1]{} 

	\numberwithin{equation}{section} 

	\newcommand{\del}[1]{\partial_{#1}} 
		
	\newcommand{\scC}{\mathscr{C}}
	\newcommand{\scD}{\mathscr{D}}
	\newcommand{\scE}{\mathscr{E}}
	\newcommand{\scH}{\mathscr{H}}
	\newcommand{\scJ}{\mathscr{J}}
	\newcommand{\scL}{\mathscr{L}}
	\newcommand{\scS}{\mathscr{S}}
	\newcommand{\scZ}{\mathcal{Z}}

	\newcommand{\bbR}{\mathbb{R}}
	
	\newcommand{\BLambda}{\mathbf{\Lambda}}

	\newcommand{\calP}{\mathcal{P}}

	\newcommand{\opA}{\mathcal{A}}
	\newcommand{\opB}{\mathcal{B}}
	\newcommand{\opb}{{b}}

	\newcommand{\norm}[2]{\| #2 \|_{#1}} 
	\newcommand{\Norm}[2]{\|\hspace{-.095em} |#2|\hspace{-.095em} \|_{#1}}	
	
	\newcommand{\AAnorm}[2]{\langle\!\langle  #2 \rangle\!\rangle_{#1}} 
	\newcommand{\AAAnorm}[2]{\langle\!\langle\!\langle  #2 \rangle\!\rangle\!\rangle_{#1}} 
	
	\newcommand{\sqgg}{\sqrt{\gamma}\,}
	\newcommand{\sqw}{\sqrt{w}\,}
	\newcommand{\sqg}{\sqrt{|g|}\,}
	
	\newcommand{\dvol}[1]{\mu_{#1}}

\begin{document}

\title[Membranes in Light cone gauge]{Local well-posedness for membranes in the light cone gauge} 
\date{\today}

\author[P.~T.~Allen]{Paul~T. Allen}
\address[Paul~T. Allen]{Albert Einstein Institute\\
 Am M\"uhlenberg 1\\
 14476 Potsdam\\ 
 Germany
\and
Interdisciplinary Arts and Sciences\\
University of Washington - Tacoma\\
1900 Commerce Street\\
Tacoma, WA 98402, USA
}
\email{paul.allen@aei.mpg.de}

\author[L.~Andersson]{Lars Andersson} 
\address[Lars Andersson]{Albert Einstein Institute\\
Am M\"uhlenberg 1\\
14476 Potsdam\\
  Germany 
}

 \email{laan@aei.mpg.de}

\author[A.~Restuccia]{Alvaro Restuccia} 
\address[Alvaro Restuccia]{Department of Physics\\
Simon Bolivar University\\
Caracas, Venezuela}
\email{arestu@usb.ve}

	
\begin{abstract} 
In this paper we consider
the classical initial value problem for the bosonic membrane in
light cone gauge.  
A Hamiltonian reduction gives a system with one constraint,
the area preserving constraint. 
The Hamiltonian evolution equations corresponding to this system, however,  
fail to be hyperbolic. 
Making use of the area preserving constraint, an equivalent system of evolution
equations is found, which is hyperbolic and has a well-posed initial value
problem.   
We are thus able to solve the initial value problem for the Hamiltonian evolution equations by means of this equivalent system.
We furthermore obtain a blowup criterion for the membrane evolution equations, and
show, making use of the constraint, that one may achieve improved regularity
estimates. 
\end{abstract}


\maketitle


\section{Introduction}
The initial value problem for the classical evolution of a physical system seeks to characterize critical points of the action functional associated to the problem in terms of an appropriate set of initial data.  Once the well-posedness of the classical field equations has been shown, the initial data not only determines the classical motion for some time interval $(0,T)$, where $T$ is the time of existence corresponding to the data, but also the Hilbert space of wave functions of the corresponding quantum-mechanical system.  
For physical systems without gauge symmetries described by conjugate pairs $(x,p)$ satisfying Hamilton's equations, the initial data is given directly by specifying the conjugate pair $(x_0, p_0)$ at an initial time.  The wave functions, in the Schr\"odinger picture, are precisely the space of functions $\varphi(x_0)$ with $\varphi \in L^2$.

In the presence of gauge symmetries the initial data is restricted by constraints and gauge-fixing conditions.  One may solve these restrictions at the classical level in terms of conjugate pairs and then determine the wave function as before, or one may consider general wave functions $\varphi(x_0)$, $\varphi \in L^2$, and restrict them at the quantum level to a subspace $\scH\subset L^2$; the domain of the quantum operators should then be dense in $\scH$.  In this latter case, one may then extend the phase space in order to realize
the BRST symmetry of the quantum system.
A key point in both procedures is to determine both the classical and quantum restrictions associated to the gauge symmetries.

The formulation of the initial value problem for gauge theories, including Einstein gravity, electromagnetism, Yang-Mills, string, and membrane theories, as well as their supersymmetric extensions, is formally solved by the approach of Dirac \cite{Di64}.\footnote{For the formal computations to be properly defined, the well-posedness of the Hamiltonian equations must be established.}
It determines in a constructive and systematic manner, the constraints associated to the gauge theory.  
Moreover the procedure ensures that the constraints are preserved in time by the Hamiltonian flow, with the Lagrange multipliers associated second-class constraints determined by the conservation procedure and Lagrange multipliers associated to first-class constraints remaining as gauge-dependent variables.
The resulting Hamiltonian formulation of the gauge theory ensures that if the constraints are satisfied initially, then they are also satisfied on $(0,T)$, for some time of existence $T$.  The first-class constraints close as an algebra under a Poisson bracket constructed from the symplectic structure of the Hamiltonian formulation.  The generators of this algebra realize the gauge symmetry of the action functional, and consequently of the field equations themselves.
(For a detailed exposition of this systematic approach in the case the Einstein gravity, including a discussion of first- and second-class constraints, see \cite{IsNe80}.)

A well-known gauge choice of the above mentioned field equations is the light cone gauge (LCG).  In this gauge, one can solve the constraints, together with the gauge fixing conditions, in terms of unconstrained physical degrees of freedom.  This property of the LCG becomes very useful when proving relevant properties of the corresponding quantum theories, one example being arguments concerning the unitarity of $S$-matrix in superstring theory.  
A treatment of Einstein gravitation has been considered in 
\cite{SchSch75}, \cite{Ka75}, \cite{ArRes76}.  In \cite{ArRes76} the positivity of the energy and coupling to gauge fields was analyzed.  For a review of non-covariant gauges in gauge field theories, including the light cone gauge, see \cite{Lieb87}.

The LCG has proven particularly fruitful in classical and quantum analyses of the $D=11$ supermembrane theory \cite{BeSeTo87},\cite{Ho82},\cite{DeHoNi88}, which is also a a relevant ingredient of $M$-theory.  
The supermembrane theory has first-class constraints associated to the generators of diffeomorphisms of the world volume and the local fermionic symmetry, known as $\kappa$-symmetry, as well as second-class fermionic constraints.  
The complete set of constraints is very difficult to treat in a covariant formulation. However, the LCG allows an explicit solution of all first- and second-class constraints.  It is also convenient in the case of both membrane and supermembrane theory to fix all symmetries up to area-preserving diffeomorphisms, which in this case are in fact symplectomorphisms (with respect to a symplectic structure defined as part of the gauge choice).\footnote{In even dimensions different from two, these symplectomorphisms are volume-preserving, but not all volume-preserving diffeomorphisms are symplectomorphisms.}
The resulting Hamiltonian, in LCG with these residual symmetries, may be analyzed without difficulty.  The constraint associated to the area-preserving diffeomorphisms may be interpreted as a symplectic generalization of the Gauss law, with nonlinear terms \`a la Yang-Mills, but arising from the symplectic bracket rather than the bracket of the $SU(n)$ Lie algebra \cite{MaOvRe01}.

Once the LCG is implemented in a Hamiltonian formulation of the supermembrane, and the first- and second-class constraints are solved, one obtains directly a canonical Hamiltonian reduction of the original formulation.  That is, the elimination of phase-space variables occurs in canonical-conjugate pairs.  The Hamiltonian in the LCG is polynomial in the remaining variables and their derivatives and, because it is formulated in terms of physical degrees of freedom (the residual gauge symmetry may be fixed in a convenient manner), it is essentially the same Hamiltonian which appears in the path-integral formulation of the quantum theory.  This means that properties of the same potential will determine both classical and quantum aspects of the theory.  The classical and quantum stability properties of the membrane or supermembrane theories are determined by the nonlinear dependence of the potential, along the configurations for which the potential becomes zero.  

These properties have been analyzed by considering a $SU(n)$ regularization of the membrane or supermembrane theory \cite{Ho82}, \cite{DeLuNi89},\cite{DeMaNi90}.  
This regularization is itself an interesting physical model, and was the starting point in the introduction of the matrix model.  
The $SU(n)$-regularized supermembrane is a maximally supersymmetric Yang-Mills theory in $1+0$ dimensions \cite{ClHal85},
 the classical field equations being ordinary differential equations and the
 corresponding quantum problem finite dimensional.  The regularized membrane
 potential has ``valleys'' extending to infinity; the value of the potential
 at the bottom of these valleys is zero.  Thus the static solutions of the
 equations of motion with zero value of the potential are unstable.
 Nonetheless, the quantum-mechanical Hamiltonian of the membrane theory has a
 discrete spectrum, due to the structure of the valleys.  In dimensions
 greater than one, the discreteness of the spectrum of the regularized
 Hamiltonian, a Schr\"odinger operator, is determined by the behavior of the
 mean value of the potential in the sense of Mol{\v{c}}anov \cite{Mo53}, \cite{MaSh05}.  This mean value tends to infinity as one moves outward along the valleys in configuration space, thus ensuring that the operator has discrete spectrum \cite{Si83},\cite{Lu83}, \cite{GaNPR07}.  In the supermembrane case, the potential becomes unbounded from below, due to Fermionic contributions, in a manner which renders the spectrum continuous \cite{DeLuNi89}.  However, the supermembrane with central charges generated by the wrapping of the supermembrane on a compact sector of the target space has discrete spectrum \cite{BoGaRe03}.  In fact the topological condition ensuring the nontrivial wrapping, which does not modify the number of local degrees of freedom, eliminates the nontrivial configurations with zero potential.

Besides this interesting relation between classical and quantum stability properties of the membrane and supermembrane theories, there are other aspects of classical supermembrane theory which reproduce quantum $\alpha'$ effects of string theory.  
(That is, perturbative quantum effects in string theory where the perturbative parameter is the inverse of the string tension.)
The closure of the $\kappa$ symmetry in supermembrane theories with a general background metric on the target space is only possible provided the background satisfies the $D=11$ supergravity equations \cite{BeSeTo88}.  The analogous result for superstring theory arises only when $\alpha'$ quantum effects are taken into account.  Furthermore, the IIA $D$-brane action in $D=10$ may be obtained from a duality transformation from the $D=11$ supermembrane compactified on a circle.  The same $D$-brane action arises from Dirichlet strings only when quantum effects are considered \cite{Sc96},\cite{To96}.

In light of these considerations, it is natural to analyze the classical initial-value problem of the membrane and supermembrane in the LCG.  This not only provides a foundation for the study of the quantum
mechanical systems corresponding to the membrane and supermembrane in LCG, by
putting the classical theory on a firm basis but, in establishing a criterion for continuing the classical solution in time, is an important first step towards the identification and characterization of whatever singularities may develop.  This is interesting not only from a classical perspective, but also from a quantum mechanical one, as singularities are expected to play an important role in the quantum theory \cite{AtWi03}, \cite{AcGu04}.  It may also provide a framework to analyze the large $n$ limit of the regularized theories and the related stability problems, one may hope to extrapolate consequences for the quantum stability problem.

From a mathematical perspective, the problem is interesting as the membrane equations, which correspond to the supermembrane field equations after the spinor dependence (i.e., the fermionic sector) has been anhilated,
are one case of an important class of geometric wave equations.  Geometrically, membranes are timelike submanifolds  with vanishing mean curvature; the equation governing this condition is the Lorentzian analogue of the minimal submanfold equations (much as wave maps are the Lorentzian analogue of harmonic maps).
Attempts to approach the problem of existence of such submanifolds by applying techniques from the theory of  differential equations are complicated by the inherent diffeomorphism invariance of the problem; in an arbitrary coordinate system the equations are not strictly hyperbolic (i.e., wave equations).  
As in the case of the mathematical study of the Einstein field equations, this difficulty can be overcome by choosing a gauge, which eliminates (or at least reduces) the diffeomorphism freedom and yields a system of equations to which PDE theory can be applied.

In fact, there are actually two levels at which the mathematical problem of 
local existence can be posed.  As the equations governing the embedding are wave-type equations, one expects to pose an initial value problem.  The first is at the level of geometry: Given an initial spacelike submanifold and timelike vectorfield on the submanifold, can one extend the submanifold in the direction of the vectorfield such that the mean curvature of the extension vanishes?  The second is at the level of of embedding functions: Given an embedding function for an initial spacelike slice, an initial `velocity' for that function, and a gauge condition, can one find an embedding function satisfying the relevant PDE (as expressed under the gauge condition)?  In either case, one would also like to show Cauchy stability as well: that not only do local solutions exist, but that they are unique and depend continuously upon the given data.

The problem of local well-posedness at the level of geometry has been recently addressed in a very general setting in \cite{Mi08}.  At the level of the PDE, that work makes use of a variation of the harmonic gauge condition, which was first applied to the membrane problem in a Minkowski ambient spacetime by \cite{AuCh78}, and has been used extensively in the study of the initial value problem for the Einstein equations (\cite{CB52}, see also \cite{FiMar72}).  
These works provide a very satisfactory resolution to the geometric local-existence problem, as well as a rather complete solution to the local existence problem for the PDE in harmonic coordinates.  It leaves open, however, the problem of local existence of embedding functions satisfying other gauge conditions which may be better-suited for addressing questions of the lifespan of solutions (and/or singularity formation of solutions) or questions arising when considering aspects of the  quantum problem.  

In fact, relatively little is known concerning the existence of solutions to the PDE when reduced by even the simplest gauge conditions.
A Hamiltonian reduction under the partial gauge  condition that the time coordinate of the submanifold coincide with the Minkowski time coordinate was considered in \cite{Mo06} (see also \cite{Ho94}).  Under this choice of foliation, the equations for codimension-$1$ membranes reduce to a first-order system.  A large number of examples of solutions satisfying this gauge condition have been constructed; see for example \cite{HoNi93}.

In this work we address the well-posedness of the membrane field equations in
the LCG.  In particular, we show that for any non-degenerate initial data
satisfying the constraints there exists a time interval $(0,T)$, with $T$
depending on the initial data, and a unique solution to equations of motion
of the reduced (in LCG) Hamiltonian corresponding to the initial data.
Furthermore, both $T$ and the solution depend continuously on the initial
data in a suitable topology.
The result is obtained by application of the theory  of hyperbolic partial
differential equations.  

The system of field equations obtained by taking variations of the membrane action in light cone gauge is a fully non-linear system which is not, however, hyperbolic.  We therefore introduce a modified system which is quasi-linear and hyperbolic, and has the property that solutions of this modified system satisfy the constraints and field equations if these are satisfied initially.  It would be interesting if this modified system has a (super-) membrane action associated to it, as it has additional constraints which are preserved by evolution and  which should appear in any quantum formulation of such a theory.  The existence of such preserved constraints is an indication that there may be a gauge theory, with gauge symmetries generated by both constraints, which under some appropriate gauge fixing reduces to the modified system we employ.

The modified system is obtained by differentiating the LCG field equations
with respect to time, and eliminating terms which vanish due to the
constraints.  Similar techniques have been used to extract hyperbolic systems
for the Einstein and
Yang-Mills equations without gauge fixing, see \cite{Abrahams:1996hh} and
references therein. 
We are able to show the existence of solutions to the modified
equation using a standard argument based on energy estimates.  The structure
of the modified system is such that the normal procedure for obtaining energy
estimates, commuting spatial derivative operators through the equation, leads
to a loss (see commutator estimate \eqref{HighSpatialCommute}).  We recover
this loss by means of an elliptic estimate and by commuting a time derivative
through the equation.  The result of our method is a local existence result
which requires slightly more regularity than expected for arguments based on
the Sobolev embedding.  We are nevertheless able, by making use of the
constraint, to obtain an improved energy estimate for {solutions} the
original equation.  With this improved energy estimate we are able to show
that the time of existence depends on the ``classically-expected'' norm of
the initial data.

In 
this paper we work in terms of integer order Sobolev spaces only. By making use
of more sophisticated techniques, the results presented here can be improved
as far as the regularity requirements are concerned. The algebraic 
structure of the
reduced field equation, makes it interesting to ask for the optimal
well-posedness result from the point of view of the regularity of initial
data. The matrix analog of the system gives an ODE
analog of the membrane system, which has been extensively studied (see for example \cite{BorHopSch91},\cite{HopYau98},\cite{ArnEtc09}). It is
interesting to consider this from an analytical point of view as a consistent
truncation of the system, and to make use of related ideas to study the
local well-posedness of the system for rough initial data. 
   
The organization of the remainder of this paper is as follows.  In the next section we outline our notational conventions, and list a number of functions spaces and related estimates appearing in the local existence proof. 
In subsequent section \ref{S:CanonicalAnalysis}, we introduce the Lagrangian
formulation of the membrane problem and perform a Dirac-style canonical
analysis of the membrane problem, deriving the reduced equations of motion
under the light cone gauge condition.  In section \ref{S:InitialValueProblem}
we give a treatment of the initial value problem by means of the modified
system described above. The modified system is hyperbolic and well-posedness
follows along essentially standard lines, with the additional difficulty that
the system is fully nonlinear. See \cite{Ko93}, \cite{Le91} for treatments of
related problems. For completeness, we
give a self-contained proof in section \ref{S:ProofOfExistence}.  
Finally in section \ref{S:Improvement} we derive the improved energy estimate, which gives us the improved estimate for the time of existence.

\section{Preliminaries}
We consider $3$-dimensional submanifolds $M$ of $D$-dimensional Minkowski space $\bbR^{D}$ with $M \cong \bbR \times \Sigma$ and $\Sigma$ some compact $2$-manifold.

We make use of a number of index sets: Greek indices $\mu,\nu, \dots$,
ranging over $0, \dots, D-1$, refer to Cartesian coordinates in Minkowski
space.  Middle Latin indices $m,n,\dots$, ranging $1, \dots, D-2$, refer to
Cartesian coordinates on Euclidean space $\bbR^{D-2}$.  Lower-case early Latin indices
$a,b,\dots,$ take values $1,2$ and refer to coordinates on $\Sigma$.
Upper-case Latin indices $A,B,\dots$,
take values $0,1,2$ and refer to coordinates on $M$.  We now describe the various coordinates in more detail and indicate which metrics are used to raise/lower each set of indices.  In all cases we sum over repeated indices.

In Cartesian coordinates $(x^\mu)$ the Minkowski metric $\eta_{\mu\nu}$ is given by 
\begin{equation}
\eta_{\mu\nu}dx^\mu dx^\nu = -(dx^0)^2 + (dx^1)^2 + \dots +(dx^{D-1})^2. 
\end{equation}
Greek are raised/lowered using $\eta_{\mu\nu}$ so that $x_\mu = \eta_{\mu\nu}x^\nu$.

We also make use of null coordinates $x^{\pm} =  \tfrac{1}{\sqrt{2}\,} ( x^0 \pm x^{D-1} )$ for Minkowski space, denoting by $(x^m)$, the remaining coordinates $(x^1, \dots, x^{D-2})$ in Euclidean space $\bbR^{D-2}$.   
With respect to these coordinates the Minkowski metric is given by $\eta_{++} = \eta_{--}=0$, $\eta_{+-} = \eta_{-+} =-1$, and $\eta_{mn} = \delta_{mn}$ is the Euclidean metric.

Let $\tau$ be some global coordinate whose level sets $\Sigma_\tau$ foliate $M$.  
When there is no confusion, we denote $\Sigma_\tau$ by $\Sigma$.
On each $\Sigma$ we use ($\tau$-independent) local coordinates $(\sigma^a)$ with $a=1,2$.  
Together $(\tau,\sigma^a) = (\xi^{A})$ give coordinates on $M$.  

Any (embedded) submanifold $M\subset \bbR^D$ is determined by the function $x = (x^\mu):M \to \bbR^D$ which induces a metric $g_{AB}$ on $M$ which in local coordinates is given by $g_{AB} = \eta_{\mu\nu}\del{A}x^\mu \del{B}x^\nu$; here $\del{A}x^\mu = \frac{\del{}x^\mu}{\del{}\xi^A}$.
In what follows we restrict attention to the case where the induced metric $g_{AB}$ is Lorentzian with $\del\tau$ a timelike direction.
The metric $g_{AB}$ induces a volume form $\dvol{g}$, given in coordinates by $\dvol{g} = \sqg d\xi^0 \wedge d\xi^1 \wedge d\xi^2$ with
$\sqg = \sqrt{-\det{[g_{AB}]}\, }$, used below to define the Lagrangian action for the membrane system.  

We denote by $\gamma_{ab}$ the metric on $\Sigma$ induced by $g_{AB}$; we require $\gamma_{ab}$ to be Riemannian.  Note that 
\begin{equation}
 \gamma_{ab} =\eta_{\mu\nu}\del{a}x^\mu \del{b}x^\nu = -(\del{a}x^+\del{b}x^- + \del{a}x^-\del{b}x^+) + \del{a}x^m\del{b}x_m.
\end{equation}
Where we desire to indicate explicitly the $x$-dependence of $\gamma_{ab}$ we write $\gamma(x)_{ab}$.
The metric $\gamma_{ab}$ gives rise to a volume form $\dvol{\gamma} = \sqrt{\gamma}\,d\sigma^1\wedge d\sigma^2$; here $\gamma$ is the determinant of $\gamma_{ab}$.
By the usual formula for the inverse of a matrix, the inverse $\gamma^{ab}$
of $\gamma_{ab}$ can be expressed 
\begin{equation}
 \gamma^{ab} = \frac{1}{\gamma}\in^{ac}\in^{bd}\gamma_{cd}.
\end{equation}
Here $\in^{ab}$ is the anti-symmetric symbol with two indices.  
(Explicitly, $\in^{12} = -\in^{21}=1$, $\in^{11} = \in^{22} =0$.)
Using the anti-symmetric symbol, the determinant $\gamma$ can be 
written as 
\begin{equation}\label{DetGamma}
\gamma = \frac{1}{2} \in^{ab}\in^{cd} \gamma_{ac} \gamma_{bd}. 
\end{equation}

\subsection{Symplectic structure}
The light cone gauge condition requires the choice of a fixed ($\tau$-independent) area form $\sqw d\sigma^1\wedge d\sigma^2$ on $\Sigma$.  
We presume $\sqw d\sigma^1\wedge d\sigma^2$ to be the area element arising from some fixed
background metric $w_{ab}$ on $\Sigma$; thus $\sqw = \sqrt{\det{w_{ab}}}$.  We may without loss of generality assume that $w_{ab}$ is real analytic.

The area form gives rise to a symplectic structure on $\Sigma$.  In local coordinates, the Poisson bracket associated to this symplectic structure is given by
	\begin{equation}
	\{ f, g\} = \frac{\in^{ab}}{\sqw} \del{a}f\, \del{b}g
	\end{equation}
where $f,g$ are functions on $\Sigma$.
Recall that the Poisson bracket is bilinear, skew ($\{ f, g\} = - \{ g, f\}$), and satisfies the Jacobi identity
	\begin{equation}\label{Jacobi}
	 0 = \{ f, \{ g, h\} \} + \{ g, \{ h , f\} \} + \{ h, \{ f , g \} \}.
	\end{equation}
In what follows we make use of the area form $\sqw d\sigma^1\wedge d\sigma^2$ when integrating on $\Sigma$, but denote the form simply by $\sqw$ (i.e. the expression $d\sigma^1\wedge d\sigma^2$ is to be implicitly understood).
Note that Stokes' theorem implies
\begin{equation}
\int_\Sigma \{f,g\}h\sqw = -\int_\Sigma f \{h,g\}\sqw.
\end{equation}

 \subsection{Derivatives and function spaces}
Here we outline our notational conventions regarding derivatives and introduce some function spaces used in the local existence results below.  
These spaces are defined with respect to a fixed atlas of coordinate charts on $\Sigma$.  

For function $v$ defined on $[0,T]\times\Sigma$, we denote by $Dv$ any coordinate derivative $\del{a}v$.  By $D^lv$, with $l$ a non-negative integer, we mean an arbitrary combination of $l$ coordinate derivatives of $v$.  We furthermore denote by $\del{}v$ any of $\del{a}v$ or $\del\tau v$.  In a slight abuse of notation, in the presence of norms we implicitly sum over all coordinate derivatives, so that 
\begin{equation}
 |\del{}v|^2 = |\del\tau v|^2 + |Dv|^2;
\end{equation}
$|D^lv|$, $|D\del{}v|$, etc. are defined analogously.

The following is an overview of the norms and function spaces used.
 \begin{description}
\item[Sobolev spaces]  Denote by $H^l$ the Sobolev space of functions on $\Sigma$ whose derivatives (in a fixed coordinate atlas) of up to order $l$ are square-integrable.  Thus
\begin{equation}
\norm{L^2}{v}^2 = \int_\Sigma |v|^2\sqw,\quad \norm{H^1}{v}^2 =  \norm{L^2}{v}^2 + \int_\Sigma |Dv|^2\sqw, \quad \text{ etc.}
\end{equation}

\item[$L^\infty$ spaces]  The following norms control the (essential) supremum of functions:
\newcommand{\esssup}{\mathop{\mathrm{ess\,sup}}}
\begin{equation}
\norm{L^\infty}{v} = \esssup_\Sigma |v|.
\end{equation}
When $v\in C^0(\Sigma)$ one can replace $\esssup_\Sigma$ with $\max_\Sigma$.  Furthermore define
\begin{equation}
\norm{W^{1,\infty}}{v} = \norm{L^\infty}{v} + \norm{L^\infty}{Dv}.
\end{equation}

\item[Curves in function spaces]  The set of maps $w:[0,T]\to H^l$ which are $r$ times differentiable with respect to $\tau$ is denoted $C^r([0,T];H^l)$ and given the norm 
$$\norm{C^r([0,T];H^l)}{w}=\sup_{[0,T]}\sum_{i=0}^r\norm{H^l}{\del\tau^iw}.$$

For functions $w\in C^r([0,T];H^l) $, we use a subscript to denote restriction to $\tau =0$.  It is important that this restriction is made \emph{after} any differentiation, so that (for example) $|\del{} w_0|^2 = |\del\tau w(0)|^2 + |Dw(0)|^2$.

\item[Spacetime norms]
We make use of two particular norms for curves of $H^l$ functions.
For
\begin{equation}
v\in C^l_T := \bigcap_{i=1}^2 C^i([0,T];H^{l-i})
\end{equation}
we make use of the norms
\begin{equation}
\norm{l}{v}^2 = \sum_{i=0}^2\norm{H^{l-i}}{\del\tau^iv}^2
\end{equation}
and
\begin{equation}
 \Norm{l,\tau}{v} = \sup_{[0,\tau]}\norm{l}{v(\cdot)}.
\end{equation}
Note that $\norm{l}{v}$ is equivalent to $\norm{H^l}{v} + \norm{H^{l-2}}{\del{}\del\tau v}$.
We also use the norms
\begin{equation}
 \AAnorm{l}{x}^2 = \norm{H^l}{x}^2 + \norm{H^l}{\del\tau{x}}^2
\end{equation}
and
\begin{equation}
\AAAnorm{l,\tau}{x} =  \sup_{[0,\tau]}\AAnorm{l}{x}
\end{equation}
for $x\in C^1([0,T];H^l)$.
\end{description}

\subsection{Results from analysis}
We make use of the following estimates, which can be proven using classical methods of calculus; see, for example, Chapter 13 \S3 of \cite{Tay97}.  
Note that while versions of these estimates hold in all dimensions, as presented here the estimates are dependent on the dimension of $\Sigma$ being $2$.
Here, and in the application of these estimates below, $C$ is a constant independent of the function(s) being estimated (and unless otherwise specified depends only on $\Sigma$, our coordinate charts, $\sqw$, and the number of derivatives being estimated).
\begin{description}
\item[Sobolev inequality] For $l>1$ and $v\in H^l$ we have $v\in C^0$ and
\begin{equation}\label{Sobolev}
\norm{L^\infty}{v} \leq C \norm{H^l}{v}.
\end{equation}

\item[Product estimate]  For $v,w\in L^\infty\cap H^l$ we have
\begin{equation}\label{BasicSobolevProduct}
\norm{H^l}{vw} \leq C\left( \norm{L^\infty}{v}\norm{H^l}{w} + \norm{L^\infty}{w}\norm{H^l}{v}\right) 
\end{equation}
Note that this estimate, together with \eqref{Sobolev} implies that $\norm{l}{xy} \leq C \norm{l}{x}\norm{l}{y}$ provided $l>1$.

\item[Elliptic regularity]  For uniformly elliptic operator $L=\del{a}[a^{ab}\del{b}(\cdot)] + b^a\del{a}$ with $a\in W^{1,\infty}$, $b^a\in L^\infty$, there exists a constant $C$, depending on the norm of the coefficients and the ellipticity constant, such that
\begin{equation}\label{EllipticRegularity}
\norm{H^2}{v} \leq C\left(\norm{L^2}{v} + \norm{L^2}{ Lv} \right)
\end{equation}

\item[Gagliardo-Nirenberg-Moser estimate]
For $l <k$ we have
\begin{equation}\label{GN}
\norm{L^{2k/l}}{D^lv} \leq C \norm{L^\infty}{v}^{1-\frac{l}{k}}\norm{L^2}{D^kv}^{\frac{l}{k}}
\end{equation}
An easy consequence of H\"older's inequality and \eqref{GN} is the following product estimate for $u,v\in H^2$
\begin{equation}\label{LowProductEstimate}
\norm{L^2}{(Du)(Dv)} \leq C\norm{H^2}{u}\norm{H^2}{v}.
\end{equation}

\item[Commutator estimate] 
The estimate \eqref{GN} also implies the following commutator estimate
\begin{equation}\label{CommutatorEstimate}
\norm{H^l}{[D^j,v]w} \leq C\left(\norm{L^\infty}{w}\norm{H^{l+j}}{v} + \norm{L^\infty}{Dv}\norm{H^{l+j-1}}{w} \right).
\end{equation}

\end{description}

\section{Canonical analysis and the light cone gauge}\label{S:CanonicalAnalysis}
We now perform an ADM-style canonical analysis of the membrane system.   Beginning with a Lagrangian action integral, we give a brief treatment for the membrane system in general before analyzing in detail the reduction under the light cone gauge condition.

\subsection{Lagrangian formulation}
We seek an embedding 
	\begin{equation}
	x = (x^\mu):  M^{1+2} \to \bbR^{D}, 
	\end{equation}
 which is critical with respect to the action given by the induced volume element
 	\begin{equation}\label{Action}
	 \scS = - \int _M \dvol{g} .
	\end{equation}

In string theory \eqref{Action} corresponds to the Nambu-Goto action.  An equivalent  formulation, expressed in terms of a polynomial Lagrangian density and corresponding to the Polyakov action of string theory, is given by
	\begin{equation}\label{FAction}
	\scS_P = -\frac12\int _M \left[ \eta_{\mu\nu} g^{AB} \del{A} x^\mu \del{B} x^\nu -1 \right]\dvol{g},
	\end{equation}
where the inverse metric $g^{AB}$ is treated as an independent variable.  Both \eqref{Action} and \eqref{FAction} lead to the same Hamiltionian formulation.

Critical points of the functional \eqref{Action} give rise to submanifolds $M\subset \bbR^D$ with vanishing mean curvature; the Euler-Lagrange equations for \eqref{Action} are
	\begin{equation}\label{E-L}
	\sqg \Box_g x^\mu = \del{A} \left[ \sqg g^{AB} \del{B} x^\mu \right] =0, \quad \mu = 0, \dots D-1,
	\end{equation}
which can be seen by computing 
	\begin{equation}
	 \delta \sqg = \frac12 \sqg g^{AB} \, \delta g_{AB} = \sqg g^{AB}  \del{A} x^\mu \del{B} [\delta x_\mu].
	\end{equation}
When $\sqg \neq 0$, the system \eqref{E-L} can be expressed (in any coordinates) as
	\begin{equation}
	\left( \delta_{\mu\nu}- g^{CD}\del{C}x_\mu \del{D}x_\nu \right)g^{AB}\del{A}\del{B} x^\nu =0.
	\end{equation}

One can interpret \eqref{E-L} as evolution equations for $x(\tau): \Sigma \to \bbR^D$.  As the induced metric $g_{AB}$ is Lorentzian, one expects the system to be hyperbolic (i.e., a wave-type equation) and thus to be able to pose the following initial value problem: For initial embedding $x_0$, and initial velocity $u_0$ does there exists an interval $[0,T]$ on which $x(\tau)$ is defined and satisfies $x(0)=x_0$, $\del\tau x(0)=u_0$, and \eqref{E-L}?

As expressed above, the evolution equations for $x^\mu$ are not strictly hyperbolic; this is due to the diffeomorphism invariance of \eqref{Action} under re-parameterizations of $M$.  (Note that it is also invariant under choice of coordinates in the Minkowski spacetime, but we have fixed these degrees of freedom.)  Thus we turn to the issue of gauge choice by performing a canonical analysis.

\subsection{Canonical (Hamiltonian) analysis}
The Hamiltonian approach plays a fundamental role in the classical and quantum analysis of field theories.  The formulation for gauge theories briefly described below was developed by Dirac in \cite{Di64}.  For a detailed description of this approach applied to classical covariant field theories see \cite{IsNe80}; for a description of path-integral quantum analysis in the presence of general constraints see \cite{Sen76}.

The starting point, as proposed by Dirac \cite{Di64}, is an action integral expressed in terms of a Lagrangian density $\scL$, defined on some time-foliated manifold $\bbR\times\Sigma$.  The density $\scL$ depends on some number of independent fields $\phi$ and their spatial derivatives up to order $k$, $D\phi, \dots, D^k\phi$, as well as on the derivatives of the first time derivative of $\phi$, $\del{t}\phi, D\del{t}\phi, \dots, D^l\del{t}\phi$.

One then introduces the Hamiltonian density $\scH$ via a Legendre transformation
	\begin{equation}
	\int_\Sigma\scH = \int_\Sigma \left( \pi \del{t}\phi - \scL\right)
	\end{equation}
where the conjugate momenta $\pi$ associated to $\phi$, defined as
	\begin{equation}
	\pi = \frac{\delta\scL}{\delta(\del{t}\phi)},
	\end{equation}
appear in the functional derivative of the Lagrangian action $L = \int_{\bbR\times\Sigma}\scL$ with respect to independent variations $\delta\phi$ and $\delta(\del{t}\phi)$:
	\begin{equation}
	\delta L = \int_{\bbR\times\Sigma} \left(  \frac{\delta\scL}{\delta \phi}\delta\phi + \frac{\delta\scL}{\delta(\del{t}\phi)}\delta(\del{t}\phi)\right).
	\end{equation}

The canonical variables $\phi$, $\pi$ take values in an infinite-dimensional manifold referred to as the phase space of the theory, which is equipped with a Poisson structure given (as a density) by	
	\begin{equation}
	[ u, v]_\calP =  \frac{\delta u}{\delta \phi}\frac{\delta v}{\del{}\pi}
		- \frac{\delta v}{\delta \phi}\frac{\delta u}{\delta \pi}.
	\end{equation}
The symplectic structure determined by the Poisson bracket plays a fundamental role in both the classical analysis of the field theory, as well as in the canonical quantization of the theory.  It is also the main algebraic structure in the deformation quantization approach \cite{BaFFL78}, \cite{Fe94}, \cite{Kon03}.
	
In gauge theories, $\del{t}\phi$ cannot be expressed in terms of unconstrained momenta as the Hessian of $\scL$ with respect to $\del{t}\phi$ becomes singular; thus there are constraints on the phase space for $(\phi, \pi)$.  Further constraints, or restrictions on the Lagrange multipliers associated to these constraints, may arise upon imposing the requirement that the vanishing of the constraints be preserved by evolution under the Hamiltonian flow (see below).  
These constraints are of two types.  First-class constraints are those which commute, on the constraint submanifold (defined to be the submanifold where all first and second class constraints are satisfied),
 with all other constraints; 
the Lagrange multipliers associated to these first-class constraints are gauge-dependent fields and remain undetermined during the canonical analysis.  Second-class constraints do not commute (on the constraint submanifold) with all other constraints and the associated Lagrange multipliers are determined by the condition that the constraints be preserved.

Notice that there is an implicit assumption in the Dirac approach concerning the structure of the constraints: they must be regular.  A point $x_0$ is a regular point of $\phi:X\to Y$ if $\phi'(x_0)$ is onto.  The constraint $\phi=0$ is regular when each point of the constraint submanifold $\{x : \phi(x)=0\}$ is a regular point of $\phi$.
Irregular constraints must be treated in a separate manner as the usual theory for Lagrange multipliers assumes regular constraints.

Once the Hamiltonian $H = \int_\Sigma \scH$ and the constraints have been determined, one can reformulate the action integral in terms of the canonical fields
	\begin{equation}\label{HamAction}
	\scS[\phi,\pi] = \int_{\bbR\times\Sigma} \left( \pi \del{t}\phi  - \scH - \lambda_aC^a\right), 
	\end{equation} 
where $C^a$ define the constraints on the phase space and $\lambda_a$ are the associated Lagrange multipliers.  The canonical Hamiltonian density defined by
	\begin{equation}
	H_c=\int_\Sigma\scH_c =\int_\Sigma  \scH + \lambda_a C^a
	\end{equation}
determines the evolution of the canonical fields via the Hamiltonian field equations
	\begin{equation}\label{HamFieldEqns}
	\begin{aligned}
	\del{t}\phi &= \left[ \phi, H_c\right]_\calP = \frac{\delta\scH_c}{\delta\pi} \\ 
	\del{t}\pi &= \left[\pi,H_c\right]_\calP= -\frac{\delta\scH_c}{\delta\phi},
	\end{aligned} 
	\end{equation}
which are obtained by varying \eqref{HamAction}, and are equivalent to the Euler-Lagrange field equations associated to the  Lagrangian $L$.  Note that for a general quantity $F=f(\phi,\pi,t)$ one has
	\begin{equation}
	\del{t}F = \left[f,H_c\right]_\calP + \del{t}f 
	\end{equation}
along the Hamiltonian flow.

The action \eqref{HamAction} is also the starting point for the Feynman path-integral formulation of quantum field theory.  Under some assumptions on the dependence of $\scH$ on $\pi$, the path integral defined from $\scS[\phi,\pi]$ is formally equivalent to the one defined from the Lagrangian action integral.

The action integral $\scS[\phi,\pi]$ is invariant under the gauge transformations generated by the first class constraints, this can be easily seen by noting that the procedure above ensures that any first class constraint $C$ has $[C,H]_\calP=0$ on the constraint submanifold.  (An interesting feature of diffeomorphism-invariant gauge theories is that $\scH=0$, i.~e. $\scH_c=\lambda_aC^a$.)
  This gauge invariance leads to degeneracies for the classical field equations  \eqref{HamFieldEqns}; thus one typically performs a (partial) gauge fixing before proceeding to analyze the equations of motion.  There is a general method for introducing gauge-fixing terms, and the corresponding Faddeev-Popov 
terms in the path-integral formulation, without solving the constraints.  The resulting effective action becomes BRST-invariant; it may be obtained from the Hamiltonian formulation following \cite{BatFr88}, \cite{FrVil75}, or from the Lagrangian formulation following \cite{BatVil83}
Both approaches present difficulties in the presence of complicated second class constraints, such as those arising in supermembrane theory.  In that case, it is most convenient to explicitly solve the constraints at the classical level (in the light cone gauge) and then proceed to quantize the theory.  Thus we follow this latter approach in this paper.

\subsection{Canonical analysis for membranes} 
We start from the Lagrangian action \eqref{FAction}, where $x^\mu$ and $g^{AB}$ are independent fields.  Alternatively, one may start from \eqref{Action}; the resulting Hamiltonians are exactly the same.

We perform the usual ADM decomposition of $g$ with respect to the foliation
$\Sigma_{\tau}$, denoting by $\gamma_{ab}$ the metric on $\Sigma$, by $N =
\frac{\sqg}{\sqgg} = \sqrt{-g^{00}}$ the lapse, and by $N^a = N^2 g^{0a}$ the
shift vector.  We raise and lower $a,b$ with $\gamma_{ab}, \gamma^{ab}$; thus
$N_a = \gamma_{ab}N^b$, etc.  
To be explicit, the metric $g_{AB}$ and its inverse $g^{AB}$ are given by
	\begin{equation}
	[g_{AB}] = \left(\begin{array}{cc}-N^2 + N^cN_c & N_b \\ N_a & \gamma_{ab} \end{array}\right)
	 \quad 
	 [g^{AB}] = \left(\begin{array}{cc}-N^{-2} & N^{-2}N^b \\ N^{-2}N^a & \gamma^{ab} - N^{-2}N^aN^b \end{array}\right). 
	\end{equation}

Treating $x^\mu$ as canonical variables, we see that the conjugate momenta to $x^\mu$ are given by
	\begin{equation}\label{DefineP}
	p_\mu =   \frac{\sqgg}{N} \left( \del\tau{x}_\mu - N^a \del{a}x_\mu\right).
	\end{equation}
 Instead of introducing the conjugate momenta to $N$, $N_a$, and $\gamma_{ab}$, it is convenient to treat them as auxiliary fields.  Furthermore, $\gamma_{ab}$ may be eliminated as an independent field and expressed in terms of the $x^\mu$ using $\gamma_{ab} = \del{a}x^\mu\del{b}x_\mu$.

We are able to solve \eqref{DefineP} for $\del\tau{x}^\mu$ in terms of $p_\mu$
	\begin{equation}\label{Invert}
	\del\tau{x}^\mu = \frac{N}{\sqgg} p^\mu + N^a\del{a}x^\mu .
	\end{equation}
The canonical Hamiltonian density is therefore given by
	\begin{align}
	\scH  &= \frac12 \frac{N}{\sqgg} \left(p^2 + \gamma \right) + N^ap_\mu \del{a}x^\mu.
	\end{align}
Here  $p^2 = p_\mu p^\mu$.

There are two constraints
	\begin{align}
	\Phi &=  \frac12  \left(p^2 + \gamma \right) \label{H-Constraint} =0\\
	\Phi_a &= p_\mu \del{a}x^\mu =0 \label{M-Constraint}
	\end{align}
with Lagrange multipliers
	\begin{align}
	\lambda = \frac{N}{\sqgg} \qquad \text{ and } \qquad
	\lambda^a = N^a.
	\end{align}
The action can now be written
	\begin{equation}\label{HamiltonianAction}
	 \scS = \int_M p_\mu \del\tau{x}^\mu - \int_M \left(\lambda \Phi + \lambda^a\Phi_a \right).
	\end{equation}
Note that $\Phi_a$ is given independently of the target (Minkowski) metric, while $\Phi$ depends on the ambient metric $\eta$.
	
The constraints $\Phi$, $\Phi_a$ are first-class constraints in the sense of Dirac.  
The quantities $\Phi$, $\Phi_a$ are the generators of time and spatial diffeomorphisms, respectively.  

One can compare the structure of this Hamiltonian to that arising in the theory of general relativity (see, for example \cite{IsNe80}); the Hamiltonian has the same linear structure (i.e., it is linear in $\Phi$, $\Phi_a$) since both are invariant under diffeomorphisms.  We also point out that $\Phi$ has the same quadratic dependence on the momentum as the corresponding time generator in general relativity.

\subsection{Light cone gauge}
We consider a (partial) gauge fixing of the above action: the \emph{light cone gauge}.  It has the property that the gauge fixing procedure gives rise to a canonical reduction of the above action, and is also the only known gauge where the $\kappa$-symmetry constraints of the supermembrane can be explicitly solved.

In order to specify the light cone gauge, we make use of the null coordinates $(x^+, x^-, x^m)$ in Minkowski space and also the ($\tau$-independent) volume form $\sqw$ on $\Sigma$. 

The light cone gauge condition for the membrane Hamiltonian is determined by taking
	\begin{equation}\label{LCGauge}
	x^+ = -p_-^0 \tau \qquad \text{ and } \qquad
	p_- = p_-^0 \sqw,
	\end{equation}
where $p_-^0$ is some constant.\footnote{The constant $p_-^0$ is related to the total momentum (in the $x^-$ direction) $P^0_- = \int_{\Sigma} p_-$  by the relation $p_-^0 = P_-^0 / \text{vol}(\Sigma)$.  Here the volume is measured with respect to $\sqw$.}  As is made evident below, this is only a partial gauge fixing, the resulting system being invariant under diffeomorphisms which are area-preserving with respect to $\sqw$. 
Note that under this gauge choice, the metric $\gamma_{ab} = \del{a}x^m\del{b}x_m$; i.e., it does not depend on derivatives of $x^{\pm}$.

We now proceed to construct the (partial) gauge-fixed Hamiltonian in light cone gauge.  The constraint $\Phi=0$ may be solved algebraically for $p_+$ in terms of $x^m$ and $p_m$
	\begin{equation}
	p_+ = \frac{1}{c\sqw} \frac12 \left(p_mp^m + \gamma \right),
	\end{equation}
while the constraint $\Phi_a=0$ determines $x^-$ (in terms of $x^m$ and $p_m$) via the relation
	\begin{equation}
	\del{a}x^- = -\frac{1}{c\sqw} p_m \del{a}x^m,
	\end{equation}
provided the integrability condition
	\begin{equation}\label{Integrability}
	\int_{C} \frac{p_m}{\sqw} dx^m =0 
	\end{equation}
holds for all closed curves $C$ in $\Sigma$.
The conjugate pairs $x^+,p_+$ and $x^-,p_-$ are thus eliminated provided this condition holds.

The Poisson bracket analysis of \eqref{Integrability} shows that it is a first-class constraint generating area-preserving diffeomorphisms and is equivalent to the local constraint
	\begin{equation}\label{AreaPreserving}
	\left\{ \frac{p_m}{\sqw}, x^m \right\} =0
	\end{equation}
in combination with the global constraint
	\begin{equation}\label{HomIntegrability}
	\int_{\scC_i} \frac{p_m}{\sqw} dx^m =0
	\end{equation}
on some basis $\{\scC_i\}$ of the homology of $\Sigma$.  
Together, \eqref{AreaPreserving}-\eqref{HomIntegrability} generate area-preserving diffeomorphisms homotopic to the identity.  The distinction between them is the following.  
The left side of \eqref{AreaPreserving} generates area-preserving diffeomorphisms under the Poisson bracket, with infinitesimal parameter $\epsilon$, a time-dependent single-valued function on $\Sigma$, i.e.
	\begin{equation}
	\delta_\epsilon x^m = \left[x^m, 
	\int_\Sigma \epsilon\left\{\frac{p_n}{\sqw}, x^n\right\}\sqw \right]_\calP  = \{x^m, \epsilon\}.
	\end{equation}
The left side of \eqref{HomIntegrability} generates area-preserving diffeomorphisms with infinitesimal parameter $\hat{\epsilon}$, where $d\hat{\epsilon}$ is a harmonic $1$-form on $\Sigma$.  Thus we may write
	\begin{equation}
	\delta_{\hat\epsilon}x^m = \left[ x^m, \hat\epsilon_i \int_{C_i}\frac{p_n}{\sqw}dx^n \right]_\calP = \{ x^m, \hat\epsilon\},
	\end{equation}
where $\hat\epsilon_i$ are time dependent functions such that $d\hat\epsilon = \hat\epsilon_i\omega^i$ for some basis $\omega^i$ of harmonic $1$-forms on $\Sigma$ normalized with respect to the homology basis $\{\scC_i\}$.  The constraint \eqref{HomIntegrability} is used in the quantum theory as a matching level condition.

We may now determine $N$ and $N^a$ so that the gauge conditions are preserved under evolution in $\tau$.  
Ensuring that $x^+ = -c\tau$ by requiring that
	\begin{equation}
	\left[  x^++c\tau , \int_{\Sigma}(  \lambda \Phi + \lambda^a\Phi_a ) \right]_\calP +c = 0
	\end{equation}
we find that 
	\begin{equation}\label{LC:Lapse}
	 N = \frac{\sqgg}{\sqw}.
	\end{equation}
In order that the second light cone gauge condition be preserved, we require that
	\begin{equation}\label{PreserveP}
	\left[ p_- - c\sqw , \int_{\Sigma}(  \lambda \Phi + \lambda^a\Phi_a ) \right]_\calP = 0  
	\end{equation}
which leads to the condition 
	\begin{equation}\label{LC:Shift}
	 \del{a}[\sqw N^a] =0.
	\end{equation}
	
Let $\Omega = \frac12\sqw \in_{ab}N^a d\sigma^b$; the condition \eqref{LC:Shift} is equivalent to $d\Omega =0$, i.e. that $\Omega$ is closed.  Consequently, on each contractible domain in $\Sigma$, we have that $\Omega$ is (locally) exact: $\Omega = df$ for some locally defined function $f$.  The one form $\Omega$ may be globally decomposed, uniquely, into its harmonic and exact parts
\begin{equation}
\Omega = h + d\Lambda_{\text{exact}}.
\end{equation}
In any local coordinates, $\Omega = \left(h_a + \del{a}\Lambda_{\text{exact}} \right)d\sigma^a$.  Thus 
\begin{equation}
N^a = \frac{\in^{ab} }{\sqw}\left( h_b + \del{b}\Lambda_{\text{exact}}\right).
\end{equation}
The form $h$ may be expressed as $h = d\Lambda_{\text{harmonic}}$ for some multi-valued function $\Lambda_{\text{harmonic}}$; write $\Lambda = \Lambda_{\text{harmonic}} + \Lambda_{\text{exact}}$.  Note that while $\Lambda$ is multi-valued, $d\Lambda$ is a well-defined geometric object.  It is also useful to keep in mind that $dx^m$ are exact one forms, as they are single-valued functions $\Sigma \to \bbR^D$.  This need not be true if the target space $\bbR^D$ is replaced by a manifold with non-trivial topology.

The light cone action may be obtained from \eqref{HamiltonianAction} by noticing that \eqref{H-Constraint}-\eqref{M-Constraint} hold and thus
	\begin{equation}
	\scS = \int_M p_\mu \del\tau{x}^\mu = \int_M \left(p_m \del\tau{x}^m - c p_+ + \del{\tau}[p_-x^-]\right),
	\end{equation}
where we may drop the last term as it is a total derivative.  The reduced Hamiltionian in light cone gauge is then given by
	\begin{equation}\label{ReducedH}
	\scH =  \sqw\left(  \frac12 \frac{p_mp^m}{\sqw^2}  + \tfrac14|\{ x, x\}|^2  \right)
	+p_m \{ \Lambda, x^m\},
	\end{equation}	
where $|\{ x, x\}|^2 =\{ x^m, x^n\}\{ x_m, x_n\} $ and we have made use of the the expression \eqref{DetGamma} for $\gamma$.

Notice that the constraints \eqref{H-Constraint}-\eqref{M-Constraint} are implemented in the action by the reduction procedure.
The last term in $\scH$ is well-defined as it is expressed in terms of $d\Lambda$.  In contrast, the expression $\Lambda \{p_m, x^m\}$ is ill-defined.  We now verify that the term $p_m \{ \Lambda, x^m\}$ indeed corresponds to a Lagrange multiplier term multiplied by the constraints.

The harmonic one form $h$ may be expressed in terms of a basis $\{\omega^i\}$
of harmonic $1$-forms\footnote{Defined eg. with respect to $w_{ab}$}.   
The basis contains $2g$ elements, where $g$ is the genus of $\Sigma$.  Thus
	\begin{equation}
	 h = \lambda_i\omega^i, \quad i=1, \dots, 2g,
	\end{equation}
where the coefficients $\lambda_i$ are functions only of $\tau$.  Making use of the bilinear Riemann identities (see, for example \cite{FarKra92})
 we have
	\begin{equation}
	 \int_\Sigma p_m \{ \Lambda, x^m\}= - \int_\Sigma \Lambda_{\text{exact}} \left\{\frac{p_m}{\sqw}, x^m\right\} \sqw + \lambda_i \int_{\scC_i} \frac{p_m}{\sqw} dx^m
	\end{equation}
and arrive at the constraints \eqref{AreaPreserving}-\eqref{HomIntegrability}.  We furthermore can interpret $\Lambda_{\text{exact}}$ and the $\lambda_i$ as the Lagrange multipliers associated to these constraints; thus our remaining gauge freedom lies in the choice of these functions.

We now turn to the equations of motion associated to the reduced Hamiltonian \eqref{ReducedH}.  
Defining $D_\tau = \del\tau + \{\cdot, \Lambda\}$, they are 
	\begin{align}
	D_\tau x^m &= \del\tau{x}^m + \{x^m, \Lambda\} = \frac{p_m}{\sqw} \\ 
	D_\tau \left(\frac{p_m}{\sqw} \right) &= \frac{\del\tau{p}_m}{\sqw} + \left\{ \frac{p_m}{\sqw}, \Lambda \right\} = \{\{x_m, x_n\}, x^n\}.
	\end{align}
Note that the constraints \eqref{AreaPreserving}-\eqref{HomIntegrability} are preserved under this evolution.	
	
This system transforms covariantly under area-preserving diffeomorphisms provided $\Lambda$ transforms appropriately.  We see that under area-preserving diffeomorphisms generated by \eqref{AreaPreserving}-\eqref{HomIntegrability}, we have
	\begin{equation}
	d \delta\Lambda = d \del\tau{\zeta} + d \{\zeta, \Lambda\} = d (D_\tau\zeta), 
	\end{equation}
where $\zeta$ is $\epsilon+\hat{\epsilon}$ as above.  The harmonic part of  $d\Lambda$ transforms as
	\begin{equation}
	\delta\lambda_i = \del\tau{\hat{\epsilon}}_i .
	\end{equation}
 The exact part of $d\Lambda$ transforms as
	\begin{equation}
	\delta \Lambda_{\text{exact}} = \del\tau{\epsilon}  + \{ \epsilon + \hat{\epsilon}, \Lambda_{\text{exact}} + \Lambda_{\text{harmonic}} \},
	\end{equation}
where we have used that $d\{  \epsilon + \hat{\epsilon}, \Lambda_{\text{harmonic}} \}$ is an exact $1$-form.

  \section{The initial value problem}\label{S:InitialValueProblem}

  The gauge freedom present in the system allows us to fix $\Lambda_{\text{exact}}$ and the $\lambda_i$.  Making the simple choice of setting each of these functions to zero, we see that Hamilton's equations of motion reduce to the second-order system (see also \cite{Ho82})
	\begin{equation}\label{Main}
	\del\tau^2{x}^m=  \left\{ \{x^m, x^n\}, x_n \right\}.
	\end{equation}
It is interesting to note that under these choices, which fix $\hat\epsilon_i$ and $\epsilon$ up to time-independent parameters, the coordinate function $\tau$ is in fact harmonic: $\Box_g\tau =0$, i.e., this the co-moving gauge.

 The remainder of this paper is devoted to studying the initial value problem for classical membranes in light cone gauge, as formulated in \eqref{Main}:
 For functions $(x_0,u_0)$ defined on $\Sigma$ we show the existence of a function $x:[0,T]\times\Sigma\to \bbR^{D-2}$ satisfying \eqref{Main} and
 	\begin{equation}\label{MainData}
	x(0)=x_0, \qquad \del\tau{x}(0) = u_0.
	\end{equation}
We require that the initial data $(x_0, u_0)$ lie in an appropriate function space, as well as satisfy an appropriate version of the constraints in order that 
	\begin{equation}\label{Constraint}
	 \{\del\tau{x}^m, x_m\}=0
	\end{equation}
be satisfied by the corresponding solution $x$.  

 \subsection{The degenerate hyperbolic system (\ref{Main})}
\label{DegenerateSystem}
The main difficulty presented by the system \eqref{Main} is that it is not strictly hyperbolic.  First, note that for fixed $x$, the operator
\begin{equation}\label{EllipticOperator}
  y \mapsto \left\{ \{y^m, x^n\}, x_n \right\} 
 = \frac{1}{\sqw}\del{a}\left[ \frac{1}{\sqw}\gamma(x)\, \gamma(x)^{ab}\del{b}y^m\right] 
\end{equation}
is elliptic with symbol $\frac{1}{w}\gamma(x)\, \gamma(x)^{ab}$.  

However, when  the operator \eqref{EllipticOperator}  is applied to $x$ itself, one must consider 
\begin{equation}
\scL: x\mapsto \left\{ \{x^m, x^n\}, x_n \right\}.
\end{equation} 
The first variation of $\scL$ is given by
\begin{equation}\label{Variation}
\delta \scL[x]y^m = \{\{y^m,x^n\},x_n\} -\{\{y^n,x_n\},x^m\} + 2\{\{x^m,x^n\},y_n\},
\end{equation}	
where we have made use of the Jacobi identity \eqref{Jacobi}.  The first term is the elliptic operator \eqref{EllipticOperator}, but the second term also contributes to the symbol and thus $\delta \scL(x)$ need not by strictly elliptic even if the metric $\gamma(x)_{ab}$ is Riemannian.  

Note that by the constraint \eqref{Constraint}, the second term in \eqref{Variation} vanishes when we take $y = \del\tau x$.  Thus the equation \eqref{Main}, together with the constraint \eqref{Constraint}, imply a non-degenerate hyperbolic equation for $u=\del\tau x$, which we use below to construct the solution $x$.  We also make use of the constraint when estimating solutions to the main equation \eqref{Main}, once their existence has been established.

Returning to the operator $\scL$, we compute in local coordinates
\begin{equation}
\begin{aligned}
 \scL(x)_m &= \{\{x_m,x^n\},x_n\}
 \\
 &= \frac{1}{w}\left(\gamma \gamma^{ab}\delta_{mn} - \epsilon^{ab}_{mn} \right)\del{a}\del{b}x^n + \text{lower order terms}
\end{aligned}
\end{equation}
where $\epsilon^{ab}_{mn}= \in^{ac}\in^{bd}\del{c}x_m\del{d}x_n$.  
Writing
\begin{equation}\label{L:Define}
L^{ab}_{mn} = \frac{1}{w}\left(\gamma \gamma^{ab}\delta_{mn} - \epsilon^{ab}_{mn} \right)
\end{equation}
we have	
\begin{equation}\label{L:Symbol}
L^{ab}_{mn}V_a^m V_b^n = \frac{1}{w}\left(\gamma |V|^2_\gamma - \left(\in^{ab}\del{a}x_mV_b^m\right)^2\right)
\end{equation}
for vector $V_a^m\in \bbR^2\times\bbR^{D-2}$.  The first term appearing in $L^{ab}_{mn}$ is diagonal and positive-definite when $\gamma_{ab}$ is Riemannian, but the second term can cause the symbol to be degenerate.
For example, consider a (local) situation with $x^1=\sigma^1$, $x^2 = \sigma^2$, $V^1_2=1$ and all other components of $x$ and $V$ zero.  

These degeneracies associated to $\scL$ prevent us from constructing solutions by direct application of standard energy methods, as the energy-type quantities associated to $L$ cannot be shown to adequately control approximate solutions.  Thus we approach the initial value problem by considering a modified system, motivated by the linearization presented above, which we now describe.

\subsection{The modified system}
Differentiating \eqref{Main} with respect to $\tau$, which has the effect of linearizing the system, and making use of the Jacobi identity \eqref{Jacobi} we obtain
  	\begin{equation}\label{DifferentiatedEqn}
	\del{\tau}^2 \del\tau{x}^m = \{\{\del\tau{x}^m, x^n\},x_n\} + 2\{\{x^m,x^n\},\del\tau{x}_n\} - \{\{\del\tau{x}^n,x_n\}, x^m\} .
	\end{equation}
When the constraint \eqref{Constraint} is satisfied, the third term on the right vanishes; the remaining terms, when viewed as an operator acting on $\del\tau{x}$, are non-degenerate.  We take advantage of this structure in the following manner.
  
For functions $x,u$ let
  	\begin{align}
	\scD &=\scD(x,u)= \{u^m, x_m\}, \\
	\scJ^m &=  \scJ^m(x,u) =\del\tau{u}^m - \{\{x^m,x^n\},x_n\}, \\
	\opA[x]u^m &= \{\{u^m, x^n\},x_n\} + 2\{\{x^m,x^n\},u_n\}.
	\end{align} 
A computation making use of the Jacobi identity and integration by parts shows that 
  	\begin{multline}
	\frac{d}{dt}\int_\Sigma \left(\scD^2 + |\scJ|^2 \right)\sqw
	= 2\int_\Sigma \Big(\scD \{ u^m, \del\tau x_m \} + \scJ_m (\del\tau^2{u}^m -\opA[x]\del\tau{x}^m ) \\
	+ (\scD-\{ \del\tau{x}^m, x_m\}) \{ \del\tau{u}^m, \del\tau x_m \}\Big)\sqw. 
	\end{multline}
Thus if $x,u$ satisfy
	\begin{equation}\label{ModMain}
	\begin{cases}
	\del\tau{x} = u, \\
	\del\tau^2{u} = \opA[x]u ,
	\end{cases}
	\end{equation}
then the conditions $\scD=0$, $\scJ^m=0$ are preserved if initially satisfied.  In particular, the function $x$ is a solution to the main equation \eqref{Main} satisfying the constraint \eqref{Constraint}.

We thus approach the initial value problem for the equation \eqref{Main} by considering the modified system \eqref{ModMain} with initial data 
	\begin{equation}\label{ModData}
	 x(0)= x_0,\quad
	 u(0)= u_0,\quad
	 \del\tau{u}(0)= u_1.
	\end{equation}
In order that the solution $(x,u)$ to the modified system \eqref{ModMain} give rise to a solution of the main equation \eqref{Main}, we require\footnote{One should also impose the global constraint condition $\oint_{C_i} u_0\cdot dx_0 =0$ on the data in order to have a solution to the Hamiltonian system; however, this condition does not play a role in our method for constructing solutions to the reduced equations \eqref{Main}.} 
\begin{align}
\scD_0 &= \{u_0^m, x^n_0\}\delta_{mn} =0, \quad \text{and}\label{DataConstraint}
\\
J^m_0 &= u_1^m -\{\{x^m_0,x^n_0\},x^l_0\}\delta_{nl} =0. \label{Compatible}
\end{align}
In the local existence result stated below, we require initial data with $x_0\in H^k$, $u_0\in H^k$, and $u_1\in H^{k-1}$.  The condition \eqref{Compatible} imposes an extra regularity condition on $x_0$.
Note, however, that due to the degeneracy discussed in \S\ref{DegenerateSystem} the condition \eqref{Compatible} does not imply $x_0\in H^{k+1}$. 	

Provided $\gamma(x)_{ab}$ is non-degenerate, the operator $\opA[x]$ is (quasi-diagonal) elliptic which in local coordinates can be written in divergence form:
\begin{equation}\label{DivergenceForm}
\begin{aligned}
\opA[x]u^m &=  \del{a}\left[\frac{\gamma(x)}{w}\,\gamma(x)^{ab}\del{b}u^m \right] \\
	&\quad + \del{a}\left[2\frac{\in^{ab}}{\sqw} \frac{\in^{cd}}{\sqw}\del{c}x^m\del{d}x^n\del{b}u_n\right]\\
	&\quad - \frac{\in^{ab}}{\sqw} \del{c}\left[\frac{\in^{cd}}{\sqw} \right]\left(\gamma_{bd}\del{a}u^m +2\del{a}x^m\del{b}x^n\del{d}u_n\right).
\end{aligned}
\end{equation}
Note that by the anti-symmetry of $\in^{ab}$ the second line does \emph{not} contain second derivatives of $u$.  
In particular, the symbol of $\opA[x]$ is $\gamma(x)\,\gamma(x)^{ab}$.
For notational convenience we denote the third line of \eqref{DivergenceForm} by  $\opb[x]u = b(x)^a\del{a}u$, and the sum of the second and third lines by $\opB[x]u=B(x)^a\del{a}u$.
Since $\gamma(x)_{ab} = \del{a}x^n\del{b}x_n$,  the operator can be schematically written in two forms,
 \begin{align}
 \begin{split}
\opA[x]u &= \del{a}\left[\gamma\gamma^{ab}\del{b}u\right] +\opB[x]u \sim D\left[(Dx)^2Du\right] \\
	& \phantom{= \del{a}\left[\gamma\gamma^{ab}\del{b}u\right] +\opB[x]u \sim}\quad+ (D^2x)(Dx)Du+ (Dx)^2Du
\end{split}
\\
&= \del{a}\left[A^{ab}\del{b}u\right]+\opb[x]u \sim D\left[(Dx)^2Du\right] + (Dx)^2Du
\end{align}
(recall $D$ represents spatial derivative); here
\begin{equation}
A^{ab} = A^{ab}_{mn} = \delta_{mn} \gamma(x)\,\gamma(x)^{ab} +2\frac{\in^{ab}}{\sqw} \frac{\in^{cd}}{\sqw}\del{c}x_m\del{d}x_n.
\end{equation}

\subsection{Well-posedness results}
  We show the existence of a solution $(x,u)$ to the modified system
  \eqref{ModMain}; by the discussion above if the constraints
  \eqref{DataConstraint}-\eqref{Compatible} are initially satisfied, this
  leads to a solution $x$ to the main equation \eqref{Main} satisfying the
  constraint \eqref{Constraint}. This procedure implies no restriction on the
  solution, in particular all solutions of the membrane initial value problem
  with the appropriate regularity can be obtained in this way. 
  
  Our approach views $x=x_0 + \int_0^\tau u$ as a functional of $u$; thus  $x$ and $u$ are required to have the same degree of spatial regularity.  
\begin{theorem}\label{LocalExistence}
Let $k\geq 4$ and let $x_0, u_0, u_1\in H^k\times H^k\times H^{k-1}$ such that $\gamma(x_0)_{ab}$ is a non-degenerate Riemannian metric, and such that \eqref{DataConstraint}-\eqref{Compatible} are satisfied.  Then we have the following.
\begin{enumerate}
\item \label{point:1} 
There exists $T>0$, depending continuously on the norm of the initial data and unique $(x,u)\in C^1([0,T];H^k)\times C^k_T$ satisfying \eqref{ModMain}-\eqref{ModData}.  
In particular, $x\in C^1([0,t];H^k)$ is a solution to \eqref{Main}-\eqref{MainData} which satisfies \eqref{Constraint}.

\item \label{point:2} 
Let $T_*$ be the maximal time of existence for $(x,u)$ as given above.  Then either $T_*=\infty$ or
\begin{equation}
\sup_{[0,T_*)}\left(\norm{L^\infty}{\gamma^{-1}} + \norm{W^{1,\infty}}{Dx} + \norm{L^\infty}{Du} \right) = \infty .
\end{equation}
\end{enumerate} 
\end{theorem}

\begin{remark} 
An adaption of standard arguments, see \cite[\S 2.3]{AML} and citations
therein, 
shows that the solution 
map $H^k \times H^k \times H^{k-1} \to C^1([0,T], H^k)$ 
given by $(x_0, u_0, u_1) \mapsto
(x,u) \in C^1([0,T], H^k)$, is continuous. Thus the initial value problem for
the system \eqref{DataConstraint}-\eqref{Compatible} is strongly well-posed. 
\end{remark} 


\section{Proof of Theorem \ref{LocalExistence}}\label{S:ProofOfExistence}

In this section we prove the existence of a solution $(x,u)$ 
as in point \ref{point:1} of Theorem \ref{LocalExistence} 
to the initial value problem for the modified system
\eqref{ModMain}-\eqref{ModData}, and establish the continuation criterion 
stated in point \ref{point:2} of that theorem. 
First, we prove 
energy estimates for the linear system
\begin{equation}\label{LinearSystem}
\del\tau^2\Phi - \opA[x]\Phi = F, \quad 
(\Phi, \del\tau\Phi)\big|_{\tau=0} = (u_0, u_1),  
\end{equation} 
associated to \eqref{ModMain}.  
Once such energy estimates have been established, a standard sequence of arguments (see \cite{ShSh98}, \cite{Maj84} for example) implies the theorem.

Below we use expressions like $\norm{L^2}{\del{}\Phi_0}$, $\norm{l}{\Phi_0}$
to denote norms calculated in terms of the initial data at $\tau = 0$. For
the higher order norms, the higher order $\tau$-derivatives are calculated
formally. 

\subsection{The linear system}
We consider first the linear system \eqref{LinearSystem} for some fixed  
$x\in C^1([0,T];H^k)$, $k\geq 4$, be such that the metric $\gamma(x)_{ab}$ is non-degenerate on $[0,T]$.
In this section we generally suppress the $x$-dependence of $\opA$, $\gamma$, and other quantities defined below.  We furthermore denote by $\Lambda$ any quantity which can be bounded by a constant times $\norm{W^{1,\infty}}{Dx}+\norm{L^\infty}{D\del\tau x}$.  Let $\Lambda_T = \sup_{[0,T]}\Lambda$.

Define the energy $E = E[\Phi]$
\begin{equation}
E(\tau) = \frac12 \int_\Sigma \left(|\del\tau\Phi|^2 + \gamma\,\gamma^{ab}\del{a}\Phi^m\del{b}\Phi_m \right)\sqw.
\end{equation}
Define $\Lambda_E=\Lambda_E[x]$ to be the smallest constant such that
\begin{equation}\label{AnormEquiv}
\Lambda_E^{-2} \norm{L^2}{\del{}\Psi}^2 \leq E[\Psi] \leq \Lambda_E^2 \norm{L^2}{\del{}\Psi}^2
\end{equation}
 on $[0,T]$ for all $v$. 
Note that $\Lambda_E$ is bounded when $\norm{L^\infty}{Dx}$ is.

Differentiating with respect to $\tau$ and integrating by parts we estimate
\begin{equation}\label{EnergyDerivative}
\del\tau E \leq \norm{L^2}{\del{\tau}\Phi}\norm{L^2}{F} +  \Lambda^2 \norm{L^2}{D\Phi}^2.
\end{equation}
Integrating \eqref{EnergyDerivative} and applying Gr\"onwall's lemma yields the standard basic energy estimate.
\begin{lemma}
If $\Phi\in C^1([0,T];L^2)\cap C^0([0,T];H^1)$ solves \eqref{LinearSystem} with $F \in L^1([0,T]; L^2)$ then
\begin{equation}\label{BEG}
 \norm{L^2}{\del{}\Phi} \leq {\Lambda_E} {e^{\int_0^\tau \Lambda^2}}
 \left(\norm{L^2}{\del{}\Phi_0} + \int_0^\tau \norm{L^2}{F}\right).
 \end{equation}
\qed
 \end{lemma}
This estimate is used below to construct estimates for $\norm{l}{\Phi}$ using the identity
\begin{equation}\label{DmixedEquation}
\del\tau^2(D^l\del\tau\Phi)-\opA(D^l\del\tau\Phi) 
	= D^l  \del\tau F 
	+ \left[D^l,\opA \right]\del\tau\Phi 
		+D^l\left(\left[ \del\tau,  \opA\right]\Phi \right) .
\end{equation}

Applying the energy estimate \eqref{BEG} to the identity \eqref{DmixedEquation} yields the following,
\begin{multline}\label{DmixedEnergyEstimate}
\norm{H^l}{\del{}\del\tau\Phi}
\leq \Lambda_E e^{\int_0^\tau \Lambda^2}\Big(
\norm{H^l}{\del{}\del\tau\Phi_0} \\
+ \int_0^\tau\big( \norm{H^l}{\del\tau F}
	+\norm{L^2}{\left[D^l,\opA \right]\del\tau\Phi}
	+\norm{L^2}{D^l\left(\left[ \del\tau,  \opA\right]\Phi\right)}
\big)
\Big).
\end{multline}

In order to obtain an energy estimate which closes, we need an estimate for spatial derivatives of $\Phi$.  This is accomplished by use  of the elliptic estimate \eqref{EllipticRegularity} which implies 
\begin{equation}
\norm{L^2}{D^l\Phi} \leq \Lambda \left(\norm{L^2}{D^{l-2}\Phi} + \norm{L^2}{\opA D^{l-2}\Phi}\right).
\end{equation}
Making use of the linear equation \eqref{LinearSystem} we have
\begin{multline}\label{ELL}
\norm{L^2}{D^l\Phi}   
\leq \Lambda \Big(
\norm{L^2}{D^{l-2}\Phi} +\norm{H^{l-2}}{F} \\
+ \norm{L^2}{D^{l-2}\del\tau^2\Phi}
+\norm{L^2}{\left[D^{l-2},\opA\right]\Phi}
\Big)
\end{multline}
Combining this with \eqref{DmixedEnergyEstimate} yields
\begin{equation}
\begin{aligned}\label{PrelimEnergyEstimate}
\norm{l}{\Phi}&\leq (1+\Lambda_T)\Lambda_Ee^{\int_0^\tau\Lambda^2}\Big(
	\norm{l}{\Phi_0} 
	+ \norm{H^{l-2}}{F(0)}
	+\norm{L^2}{\left[D^{l-2},\opA\right]\Phi}
\\
& \phantom{(1+\Lambda_T)\Lambda_Ee^{\int_0^\tau\Lambda^2}\Big(}
+\int_0^\tau \big(
	\norm{H^{l-2}}{\del\tau F}
	+\norm{l-1}{\Phi}
\\
& \phantom{(1+\Lambda_T)\Lambda_Ee^{\int_0^\tau\Lambda^2}\Big(\int_0^\tau \big(+}
	+\norm{L^2}{\left[D^{l-2},\opA \right]\del\tau\Phi}
\\
& \phantom{(1+\Lambda_T)\Lambda_Ee^{\int_0^\tau\Lambda^2}\Big(\int_0^\tau \big(+}
	+\norm{L^2}{D^{l-2}\left(\left[ \del\tau,  \opA\right]\Phi\right)}
	\big)
\Big)
\end{aligned}
\end{equation}

We estimate the commutator terms as follows.
\begin{lemma}
Suppose $\Psi\in H^2$, then
\begin{subequations}\label{AllCommute}
\begin{align}
\begin{split}\label{LowSpatialCommute}
\norm{L^2}{[D,\opA] \Psi}
&\leq \Lambda \norm{H^3}{Dx}\norm{H^2}{\Psi}
\end{split}
\\
\begin{split}\label{LowMixedCommute}
\norm{L^2}{D\left([\del\tau,\opA]\Psi\right)} 
&\leq \Lambda \AAnorm{3}{Dx}\norm{H^2}{\Psi}
\end{split}
\end{align}
If furthermore $\Psi\in H^{l}$ for $3 \leq l \leq k$, then by \eqref{Sobolev} $D\Psi\in L^\infty$ and
\begin{align}
\begin{split}\label{HighSpatialCommute}
\norm{L^2}{[D^{l-2},\opA] \Psi}
&\leq \Lambda \left( 
\Lambda \norm{H^{l-2}}{D\Psi} 
+ \norm{H^{l-1}}{Dx}\norm{L^{\infty}}{D\Psi}
\right)
\end{split}
\\
\begin{split}\label{HighMixedCommute}
\norm{L^2}{D^{l-2}\left([\del\tau,\opA]\Psi\right)} 
&\leq \Lambda\big( \Lambda \norm{H^{l-1}}{D\Psi}  + \AAnorm{{l-1}}{D x}\norm{L^\infty}{D\Psi}\big).
\end{split}
\end{align}
\end{subequations}
\end{lemma}
\begin{proof}
Direct computation shows
\begin{equation}
\norm{L^2}{[D,\opA] \Psi}\leq \norm{L^2}{(D^2A)(D\Psi)} + \norm{L^2}{(DA)(D^2\Psi)} + \norm{L^2}{(Db)(D\Psi)}.
\end{equation}
Application of \eqref{LowProductEstimate}, \eqref{BasicSobolevProduct} and \eqref{Sobolev} implies 
\eqref{LowSpatialCommute}.  Similar considerations imply \eqref{LowMixedCommute}.
To see \eqref{HighSpatialCommute} write
\begin{equation}
[D^{l-2},\opA] \Psi = \left[ D^{l-2}D, A\right]D\Psi - (DA)D^{l-2}D\Psi
+D^{l-2}\left(bD\Psi \right) - bD^{l-2}D\Psi.
\end{equation}
By the product estimate \eqref{BasicSobolevProduct} we have 
\begin{equation}
 \norm{H^{l-1}}{A} 
 \leq \Lambda \norm{H^{l-1}}{Dx},
\quad\text{ and }\quad
 \norm{H^{l}}{b} 
 \leq  \Lambda \norm{H^{l-2}}{Dx}
\end{equation}
the second of which implies
\begin{align}
 \norm{H^{l-2}}{bD\Psi} 
 & \leq \Lambda\left( \Lambda\norm{H^{l-2}}{D\Psi} 
 	+ \norm{H^{l-2}}{Dx}\norm{L^{\infty}}{D\Psi}\right).
\end{align}
The estimate \eqref{HighSpatialCommute} follows from commutator estimate \eqref{CommutatorEstimate}.
Finally, the product estimate \eqref{BasicSobolevProduct} and Sobolev inequality \eqref{Sobolev} imply \eqref{HighMixedCommute}.
\end{proof}

As an immediate application of \eqref{AllCommute} we have
\begin{equation}\label{SpatialCommuteTau}
\begin{aligned}
\norm{L^2}{\left[D,\opA\right]\del\tau\Phi}
&\leq \Lambda \norm{H^3}{Dx}\norm{H^2}{\del\tau\Phi}
\\
 \norm{L^2}{\left[D^{l-2},\opA\right]\del\tau\Phi}  
&\leq \Lambda \left(  \Lambda  \norm{H^{l-2}}{D\del\tau\Phi} 
+ \norm{H^{l-1}}{Dx}\norm{L^{\infty}}{D\del\tau\Phi} \right) .
\end{aligned}
\end{equation}
Furthermore
\begin{equation}\label{IntegrateCommute}
 \begin{aligned}
 \norm{L^2}{\left[D,\opA\right]\Phi} &\leq  \Lambda \norm{H^3}{Dx}\left(\norm{H^2}{\Phi_0} +\int_0^\tau \norm{3}{\Phi}\right)
 \\
 \norm{L^2}{\left[D^{l-2},\opA\right]\Phi} &\leq
 \Lambda^2\left(\norm{H^{l-2}}{D\Phi_0} + \int_0^\tau \norm{H^{l-2}}{D\del\tau\Phi} \right) \\
&\quad +\Lambda\norm{H^{l-1}}{Dx}\left(\norm{H^2}{D\Phi_0} + \int_0^\tau \norm{H^2}{D\del\tau\Phi}\right)
\end{aligned}
\end{equation}
where we have used the Sobolev inequality \eqref{Sobolev} in the last line.

Applying the commutator estimates above to \eqref{PrelimEnergyEstimate}, followed by application of Gr\"onwall's lemma, yields the following energy inequalities.

\begin{lemma}
A solution $\Phi$ to the linear system \eqref{LinearSystem} can be estimated for $\tau \in [0,T]$ by
\begin{multline}\label{LowLinearEnergyEstimate}
\norm{3}{\Phi} \leq  \Lambda_E(1+\Lambda_T)e^{\tau\Lambda_E(1+\Lambda_\tau)\Lambda_\tau\AAAnorm{3,\tau}{Dx}}
\Big((1+\Lambda\norm{H^{3}}{Dx})\norm{3}{\Phi_0} \\
+\norm{H^{1}}{F(0)}
+\int_0^\tau\norm{H^1}{\del\tau F} \Big)
\end{multline}
and, for $k\geq 4$, by
\begin{multline}\label{HighLinearEnergyEstimate}
\norm{k}{\Phi}\leq \Lambda_E (1+\Lambda_T)e^{ \tau\Lambda_E(1+\Lambda_\tau) \Lambda_\tau\AAAnorm{{k-1},\tau}{Dx}}
\Big(
(1+\Lambda\norm{H^{k-1}}{Dx})\norm{k}{\Phi_0} \\
+ \norm{H^{k-2}}{F(0)} + \int_0^\tau\norm{H^{k-2}}{\del\tau F} 
\Big).
\end{multline}
\qed
\end{lemma}

\subsection{Application to non-linear system}

In order to apply the above energy estimates to the non-linear system \eqref{ModMain} we fix $k\geq 4$ and $(x_0, u_0, u_1)\in H^k\times H^k\times H^{k-1}$.  For $T,R>0$ let
\begin{equation}
V^k_{R,T} = \left\{ v\in C^k_T\,:\, \Norm{k,T}{v}\leq R, v(0)=u_0, \del\tau v(0)=u_1\right\}.
\end{equation}
For any $v\in V^k_{R,T}$ let $x(v) = x_0+\int_0^\tau v$.  There exists $K_0$ such that the restriction $T<K_0/R$ implies $\Lambda_E[x]\leq 2\Lambda_E[x_0]$.  We restrict to such values of $T$; thus $\gamma[x(v)]_{ab}$ is uniformly elliptic for all $v\in V^k_{R,T}$.  

Define $\Phi=\Phi[v]$ to be the solution to
\begin{equation}\label{DefineIterate}
\del\tau^2\Phi - \opA[x(v)]\Phi =0,\quad (\Phi, \del\tau\Phi)\big|_{\tau=0} = (u_0, u_1). 
\end{equation}
The linear system \eqref{DefineIterate} is hyperbolic and it follows from our
assumptions that the coefficients are sufficiently regular that standard
existence results apply, see \cite{ShSh98}. 

Estimating
\begin{equation}
\Lambda\norm{H^{k-1}}{Dx}\leq \norm{H^{k-1}}{Dx}^2 \leq 4\left( 1+ \norm{H^{k-1}}{Dx_0}^2 + \tau^2 \Norm{H^{k-1},\tau}{Dx}\right)
\end{equation}
the energy estimate \eqref{HighLinearEnergyEstimate}, together with the Sobolev inequality \eqref{Sobolev} which provides control of $\Lambda$, implies the following.

\begin{lemma}\label{Boundedness}
There exists $R_0>0$ depending on $x_0,u_0,u_1$, such that for each $R\geq R_0$ there exists $T_R>0$ such that for all $T\leq T_R$ the map $v\mapsto \Phi[v]$ takes $V^k_{R,T}$ to itself.
\qed
\end{lemma}

Fixing some such $R$, we now show that for a possibly smaller $T>0$, the map $v\mapsto \Phi[v]$ is a contraction with respect to the $C^3_T$ norm.

\begin{lemma}
Let $R,T_R$ be as given by Lemma \ref{Boundedness}.  For a possibly smaller value of $T_R$, we have for each $T\leq T_R$ that
\begin{equation}
\Norm{3,\tau}{\Phi[v^1]-\Phi[v^2]} < \Norm{3,\tau}{v^1-v^2} 
\end{equation}
for all $v^1,v^2\in V^k_{R,T}$.
\end{lemma}
\begin{proof}
Let $\Phi^i=\Phi[v^i]$ and $x^i = x(v^i)$ for $i=1,2$.  Then
\begin{equation}
\del\tau^2(\Phi^1-\Phi^2) - \opA[x^1](\Phi^1-\Phi^2) = (\opA[x^1]-\opA[x^2])\Phi^2.
\end{equation}
Schematically,
\begin{equation}\label{ExpandRHS}
\begin{aligned}
 \del\tau\left[(\opA[x^1]-\opA[x^2])\Phi \right] 
 &= D\left[ \left(\del\tau A(x^1) -\del\tau A(x^2)\right)D\Phi\right] \\
 &\quad+ D\left[ \left( A(x^1) -A(x^2)\right)D\del\tau \Phi \right] \\
 &\quad+ \left(\del\tau b(x^1) - \del\tau b(x^2)\right)D\Phi \\
 &\quad+ \left( b(x^1) - b(x^2)\right)D\del\tau \Phi
\end{aligned}
\end{equation}
Using the mean value theorem, along with product estimate \eqref{BasicSobolevProduct} and the Sobolev inequality \eqref{Sobolev} the $H^1$ norm of each of the first two lines in \eqref{ExpandRHS} is controlled by $C(R)\AAnorm{2}{D(x^1-x^2)}$.  A direct estimate yields the same bound for the latter terms.  Applying the energy estimate \eqref{LowLinearEnergyEstimate}, one can choose $T$, depending on $R$, such that $v\mapsto \Phi[v]$ is a contraction
\end{proof}

A standard argument (see \cite{ShSh98},\cite{Maj84}) using that the solution map $v\mapsto \Phi[v]$ is bounded in $C^k_T$ and a contraction in $C^3_T$ yields part 1 of Theorem \ref{LocalExistence}.

\subsection{Continuation criterion}
We now establish the second part of Theorem \ref{LocalExistence}.
Consider the solution $(x,u)$ to \eqref{ModMain}-\eqref{ModData} and suppose $[0,T_*)$ is the maximal interval of existence.  Let
\begin{equation}
\BLambda_\tau = \sup_{[0,\tau)}\left(\norm{W^{1,\infty}}{Dx}+\norm{L^\infty}{Du}+\Lambda_E[x]\right) .
\end{equation}
Note that $\Lambda_E$ is finite if and only if $\norm{L^\infty}{Dx}$ and $\norm{L^\infty}{|\gamma(x)|^{-1}}$ are.

Using \eqref{AllCommute} and \eqref{SpatialCommuteTau} applied to \eqref{PrelimEnergyEstimate} we see for any $\tau < T_*$ that
\begin{equation}
\norm{k}{u}\leq (1+\Lambda_\tau)\Lambda_Ee^{\int_0^\tau\Lambda^2}\Big( (1+\Lambda^2)\norm{k}{u_0}
+\int_0^\tau(1+\Lambda_\tau^2)\norm{k}{u}\Big).
\end{equation}
Thus Gr\"onwall's lemma implies that 
\begin{equation}\label{FirstNLEnergyBound}
 \Norm{k,\tau}{u}\leq (1+\Lambda_\tau^2)^2\Lambda_E e^{\tau(1+\Lambda_\tau^2)^2\Lambda_E}\norm{k}{u_0}.
\end{equation}
If $\BLambda_{T_*} <\infty$ and $T_*<\infty$, then one may extend $x,u,\del\tau u$ to $[0,T_*]$ such that their restriction to $\tau=T_*$ satisfies the hypotheses of the local existence theorem.  As this contradicts the maximality of $T_*$, we obtain the second part of Theorem \ref{LocalExistence}.

\section{Improved energy estimate}\label{S:Improvement}
Having established the existence of a solution $x\in C^1([0,T];H^k)$ to the main equation \eqref{Main}, we are able to establish an improved energy estimate by making use of the constraint \eqref{Constraint}.  In particular, we obtain the following.
  
\begin{theorem}\label{Improvement} Let $k \geq 4$. 
For a solution $x \in C^1([0,T],H^k)$ \eqref{Main} with initial data $(x_0,
u_0)$, the maximal time of existence $T_*$ depends continuously on 
$\norm{H^k}{x_0}+\norm{H^{k-1}}{u_0}$. 
\end{theorem}
\begin{remark} Comparing with  theorem \ref{LocalExistence}, we see that
  regularity requirement on the 
initial data is less. In particular, in  theorem \ref{LocalExistence} it is
required that $x_0, u_0$ are in $H^k$, while in theorem \ref{Improvement},
the regularity condition is of the type usually encountered for hyperbolic
systems. Thus, from this point of view, when taking the area preserving
constraint \eqref{Constraint} into account, the degenerate hyperbolic system
\eqref{Main} behaves very much like a hyperbolic system.  
\end{remark}

 From the extension criterion, we know that $x$ may be continued so long as $\BLambda_\tau$ is finite.  The Sobolev estimate \eqref{Sobolev} implies that $\BLambda_\tau^2$ can be estimated by $\norm{H^{k-1}}{\del{}x_0}^2 = \norm{H^{k-1}}{Dx_0}^2 + \norm{H^{k-1}}{u_0}^2$.  We now establish an estimate for this quantity by deriving an energy estimate for the main equation \eqref{Main} itself.

Retaining the notation above, we have 
\begin{equation}
\Lambda = \Lambda[x] = C\left(\norm{W^{1,\infty}}{Dx} + \norm{L^\infty}{Du}\right) 
\end{equation}
where we let the constant $C$ increase (independent of $x$) as needed and $\Lambda_\tau = \sup_{[0,\tau)}\Lambda$.

Following the discussion in \S\ref{DegenerateSystem} we write \eqref{Main} as
\begin{equation}
\del\tau^2{x}_m - L_{mn}^{ab}\del{a}\del{b}x^n =F_m
\end{equation}
where $F_m = F_m(Dx)$.  In order to estimate derivatives of $x$ we make use of the identity
\begin{equation}\label{DiMain}
\del\tau^2 (D^lx) - L(D^lx) = D^l(F) - \left[D^l, L \right]x
\end{equation}	
where $L = L^{ab}_{mn}\del{a}\del{b}$ and we have dropped the $m,n$ indices for notational convenience.

Define the  energy associated to $\del\tau^2-L$ by 
  	\begin{equation}\label{Define:scE}
	\scE[y]=\scE[y](\tau) = \frac12 \int_{\Sigma}\left( |\del\tau{y}|^2 + L^{ab}_{mn}(Dx)(\del{a}y^m)(\del{b}y^n)\right)\sqw.
	\end{equation}
We are of course interested in the cases $y=D^lx$ for $l=0,\dots, k-1$.
	
From \eqref{L:Symbol} the non-degeneracy of $\gamma(x)$ implies there is a constant $\Lambda_\scE = \Lambda_\scE[x;T]$ such that on $[0,T]$
\begin{equation}
\scE[y] \leq \Lambda_\scE^2 \norm{L^2}{\del{}y}^2
\quad \text{ and }\quad
\norm{L^2}{\del{}y}^2 \leq \Lambda_\scE^2 \frac12 \int_\Sigma\left( |\del\tau{y}|^2 + \frac{\gamma}{w}|Dy|^2_\gamma\right)\sqw.
\end{equation}	
Note that since $\Lambda_\scE$ essentially controls the ellipticity of $\gamma(x)_{ab}$, it is effectively equivalent to $\Lambda_E$ above.

\textit{A priori}, the energy $\scE[y]$ does not necessarily control $\norm{L^2}{\del{}y}$.  Rather
\begin{equation}\label{FirstNormCompare}
\norm{L^2}{\del{}y}^2 \leq \Lambda_\scE^2\left(\scE[y] + \scZ[y]\right)
\end{equation}
where
\begin{equation}
\scZ[y]  = \frac12 \int_{\Sigma} \left\{ y^m, x_m\right\}^2\sqw.
\end{equation}
The following lemma shows that the ``error'' $\scZ$ can be controlled for $y=D^lx$, if $x$ satisfies the constraint \eqref{Constraint}.

\begin{lemma}
Since the solution $x\in C^1([0,T];H^k)$ satisfies the constraint \eqref{Constraint}, then for $0\leq l \leq k-1$ we have
\begin{equation}\label{EstimateZ}
\begin{aligned}
\scZ[D^{l}x] &\leq e^\tau \left( \scZ[D^{l}x](0) + \int_0^\tau \Lambda^2 \norm{L^2}{D^{l}\del{}x}^2\right)\\
	&\leq e^\tau \left( C\norm{H^{l}}{\del{}x_0}^2 + \tau \Lambda^2_\tau \Norm{L^2,\tau}{D^{l}\del{}x}^2\right).
\end{aligned}
\end{equation}
\end{lemma}
\begin{proof}
Applying $D^{l}$ to the constraint yields
\begin{equation}
\left\{D^{l} u^m, x_m\right\} = -\sum_{j=1}^{l} \left\{D^{l-j} u^m, D^{j}x_m\right\}
\end{equation}
from which we see that
\begin{equation}
\del\tau\scZ[D^{l}x] \leq \scZ[D^{l}x] + \sum_{j=1}^{l} \int_\Sigma \left\{D^{l-j}u^m, D^{j}x_m\right\}^2\sqw.
\end{equation}
When $j=1$ or $l$ we have
\begin{equation}
\int_\Sigma \left\{D^{l-j}u^m, D^{j}x_m\right\}^2\sqw \leq \norm{L^\infty}{D\del{}x}^2 \norm{L^2}{D^{l}\del{}x}^2.
\end{equation}
When $2\leq j \leq l-1$ the H\"older estimate implies
\begin{equation}
\int_\Sigma \left\{D^{l-j}u^m, D^{j}x_m\right\}^2\sqw \leq \norm{L^{2p}}{D^{l-j}Du} \norm{L^{2q}}{D^{j-1}D^2x} 
\end{equation}
with $\frac{1}{p} = 1-\frac{1}{q} = \frac{l-j}{l-1}$.  From the interpolation estimate 
\eqref{GN}
we have
\begin{align}
 \norm{L^{2p}}{D^{l-j}Du} &\leq C \norm{L^2}{D^{l-1}Du}^{\frac{1}{p}} \norm{L^\infty}{Du}^{1-{\frac{1}{p}}}
	 \\
 \norm{L^{2q}}{D^{j-1}D^2x} &\leq C \norm{L^2}{D^{l-1}D^2{x}}^{1-{\frac{1}{p}}} \norm{L^\infty}{D^2{x}}^{\frac{1}{p}}.
\end{align}
Thus 
\begin{equation}
 \del\tau \scZ[D^{l}x] \leq \scZ[D^{l}x]  +\Lambda^2 \norm{L^2}{D^{l}\del{}x}^2.
\end{equation}  
Integrating  and applying Gr\"onwall's lemma yields \eqref{EstimateZ}.
\end{proof}

 The estimate \eqref{EstimateZ} for $Z$, together with \eqref{FirstNormCompare}, implies that when
\begin{equation}\label{SmallT}
T e^T  < \Lambda_\scE^{-2}\Lambda_T^{-2} 
\end{equation}
then
\begin{align}\label{ImprovedNormCompare}
\Norm{L^2,\tau}{\del{}D^{l}x}^2 &\leq C\Lambda_\scE\left(\norm{H^{l}}{\del{}x_0}^2 +\sup_{[0,\tau]}\scE[D^{l}x]^{}\right).
\end{align}
This facilitates the construction of the following energy estimate.

\begin{lemma}
When \eqref{SmallT} holds
the first derivatives $\del{}x$ of solution $x\in C^1([0,T];H^k)$ can be controlled in $L^\infty([0,T];H^l)$ for $0\leq l \leq k-1$ by
\begin{equation}\label{ImprovedEnergyBound}
 \Norm{H^l,\tau}{\del{}x} \leq C\Lambda_\scE e^{\int_0^\tau \Lambda^2}\norm{H^l}{\del{}x_0}.
\end{equation}
\end{lemma}
\begin{proof}
A straightforward computation using integration by parts yields
\begin{equation}\label{BasicLEnergy}
\del\tau \scE[y] \leq \norm{L^2}{\del{\tau}y}\norm{L^2}{\del\tau^2 y - Ly} + \Lambda^2 \norm{L^2}{Dy}^2.
\end{equation}
In order to apply this to $y=D^{l}x$ we must estimate the right side of \eqref{DiMain}, which we write schematically as
\begin{equation}
 D^lF - \left[D^l, L \right]x = D^lF - \left[D^lD, L^{ab}_{mn}\right] Dx + D^l\left( (DL^{ab}_{mn})(Dx)\right).
\end{equation}
Applying the commutator estimate \eqref{CommutatorEstimate} and product estimate \eqref{BasicSobolevProduct} we have 
\begin{equation}
\norm{L^2}{D^lF - \left[D^l, L \right]x}  \leq 
\Lambda^2\norm{H^{l}}{Dx}.
\end{equation}
Thus integrating \eqref{BasicLEnergy} and using \eqref{ImprovedNormCompare} we have
\begin{equation}
\Norm{H^l,\tau}{\del{}x}^2 \leq C\Lambda_\scE^2 \Big(\norm{H^l}{\del{}x_0} 
+\int_0^\tau  \Lambda^2\norm{H^{l}}{Dx}^2 
\Big)
\end{equation}
Applying Gr\"onwall's lemma we obtain \eqref{ImprovedEnergyBound}.
\end{proof}

We now show that there exists $T>0$, depending only on $\norm{H^{3}}{\del{}x_0}$, such that the solution $x$ to \eqref{Main} is defined on $[0,T]$.  The key is to establish estimates for $\Lambda_T$ and $\Lambda_\scE$, and to ensure that \eqref{SmallT} is satisfied.

In order to estimate $\Lambda_\scE$ note that there exists $K$, depending on $\Lambda_\scE[x_0]$, such that 
$\Lambda_\scE[x] \leq 2\Lambda_\scE[x_0]$
whenever 
\begin{equation}\label{ControlNormCompare}
 T\leq K\Lambda_T^{-1}.
\end{equation}  
For such $T$, the condition $\eqref{SmallT}$ follows from
\begin{equation}\label{ModSmallT}
Te^T\leq C\Lambda_\scE[x_0]^{-2}\Lambda_T^{-2}. 
\end{equation}
Making use of the Sobolev inequality \eqref{Sobolev} the energy estimate \eqref{ImprovedEnergyBound} implies that when \eqref{ControlNormCompare} is satisfied
\begin{equation}\label{ControlLambda}
\Lambda_T \leq C\Lambda_\scE[x_0] e^{T \Lambda^2_T}\norm{H^{k-1}}{\del{}x_0}.
\end{equation}
Since \eqref{ControlNormCompare}-\eqref{ModSmallT}-\eqref{ControlLambda} all hold with strict inequality for $T=0$, we may choose $T>0$, depending on $\norm{H^{k-1}}{\del{}x_0}$, such that they continue to hold on $[0,T]$.

 
\subsection*{Acknowledgements}
  LA and PA thank the Mittag-Leffler Institute for hospitality and support during the 2008 program \emph{Geometry, Analysis, and Gravitation}.  LA is supported in part by the NSF,
contract no.  DMS 0707306.  The work of AR was supported by a grant from the
Albert Einstein Institute. Part of this work was done while PA held a
post-doc position at the Albert Einstein Institute, and during a long-term
research visit at the same institution.


\providecommand{\bysame}{\leavevmode\hbox to3em{\hrulefill}\thinspace}
\providecommand{\MR}{\relax\ifhmode\unskip\space\fi MR }
\providecommand{\MRhref}[2]{%
  \href{http://www.ams.org/mathscinet-getitem?mr=#1}{#2}
}
\providecommand{\href}[2]{#2}

\end{document}